\documentclass[preprint, 10pt, conference, letterpaper]{IEEEtran}
\makeatletter
\def\ps@headings{%
\def\@oddhead{\mbox{}\scriptsize\rightmark \hfil \thepage}%
\def\@evenhead{\scriptsize\thepage \hfil \leftmark\mbox{}}%
\def\@oddfoot{}%
\def\@evenfoot{}}
\makeatother
\ifCLASSINFOpdf
 \usepackage[pdftex]{graphicx}
\else

\fi

\usepackage{float}
\usepackage{amsmath}
\usepackage{algorithm}
\usepackage{algpseudocode}
\usepackage{array}
\usepackage{mdwmath}
\usepackage{eqparbox}
\usepackage[mathscr]{euscript}
\usepackage{fixltx2e}
\usepackage{url}
\usepackage{amsfonts}
\usepackage[svgnames]{xcolor}

\newtheorem{proposition}{Proposition}
\newtheorem{corollary}{Corollary}

\newtheorem{lemma}{Lemma}
\newcommand{\automaton}{\mathcal{PDA}}
\newcommand{\G}{\mathcal{CFG}}
\newcommand{\States}{\mathcal{S}}
\newcommand{\Alphabet}{\Sigma}
\newcommand{\Symbols}{\Gamma}
\newcommand{\Transitions}{\delta}
\newcommand{\startsymb}{Z_0}
\newcommand{\Fstates}{\States_F}

\newcommand{\alphabet}{\mathcal{A}}
\newcommand{\trace}{\mathcal{T}}
\newcommand{\nodes}{\V}
\newcommand{\AD}{\mathcal{F}}
\newcommand{\qos}{q}
\newcommand{\nb}{nb}
\newcommand{\V}{\mathcal{V}}
\newcommand{\E}{\mathcal{E}}
\newcommand{\Hh}{\mathcal{H}}
\newcommand{\Gg}{\mathcal{G}}
\newcommand{\C}{\mathcal{C}}
\newcommand{\N}{\mathcal{N}}
\renewcommand{\P}{\mathcal{P}}

\IEEEoverridecommandlockouts

\begin{document}

%

%
%
%

\author{\IEEEauthorblockN{Mohamed~Lamine~Lamali\IEEEauthorrefmark{1}, Nasreddine~Fergani\IEEEauthorrefmark{1}, Johanne Cohen\IEEEauthorrefmark{2}, H\'elia Pouyllau\IEEEauthorrefmark{3}
\IEEEauthorblockA{\IEEEauthorrefmark{1}Nokia Bell Labs. France.}
\IEEEauthorblockA{\IEEEauthorrefmark{2}LRI, Univ. Paris-Sud, CNRS, Universit\'e Paris-Saclay. France.}
\IEEEauthorblockA{\IEEEauthorrefmark{2}Thales Research \& Technology. France.}
}
\url{mohamed_lamine.lamali@nokia.com }\ \ \
\url{johanne.cohen@lri.fr }\ \ \ 
\url{helia.pouyllau@thalesgroup.com}
}

\title{Path computation in multi-layer networks: Complexity and algorithms\thanks{Preprint of the paper accepted for IEEE INFOCOM $2016$.}}

\maketitle

\begin{abstract}
Carrier-grade networks comprise several layers where different protocols coexist. Nowadays, most of these networks have different control planes to manage routing on different layers, leading to a suboptimal use of the network resources and additional operational costs. However, some routers are able to encapsulate, decapsulate and convert protocols and act as a liaison between these layers. A unified control plane would be useful to optimize the use of the network resources and automate the routing configurations. Software-Defined Networking (SDN) based architectures, such as OpenFlow, offer a chance to design such a control plane. One of the most important problems to deal with in this design is the path computation process. Classical path computation algorithms cannot resolve the problem as they do not take into account encapsulations and conversions of protocols. In this paper, we propose algorithms to solve this problem and study several cases: Path computation without bandwidth constraint, under bandwidth constraint and under other Quality of Service constraints. We study the complexity and the scalability of our algorithms and evaluate their performances on real topologies. The results show that they outperform the previous ones proposed in the literature.

\end{abstract}


\section{Introduction}
Carrier-grade networks generally encompass several layers involving different technologies and protocols. To support some services, such as a Virtual Private Network (VPN), a path across network equipments must be identified and the equipments be configured accordingly. Under stringent requirements of Quality of Service (QoS \--- e.g., end-to-end delay, geographic zone avoidance, etc.), computing such a path within a single layer is not always possible. Hence, one of the key challenges is to determine the end-to-end path that uses the appropriate \textit{adaptation functions} over the protocols: The mapping from a protocol to another being realized through \textit{encapsulation} (e.g., Ethernet over IP/MPLS~\cite{RFC4448}), \textit{decapsulation} (the reverse operation) or \textit{conversion} (e.g., IPv4 to IPv6~\cite{RFC6144}) functions. Consequently, the path computation process should take into account the adaptation function capabilities of the network equipments in order to ensure \textit{path feasibility}: If a protocol is encapsulated in another one, it must be decapsulated (or unwrapped) further in the path. If several encapsulations are nested, the corresponding decapsulations must occur in the right order. Here, the multi-layer context should be taken in a broad sense: Presence of several protocols and technologies that can be nested, encapsulated, converted, etc.

Dealing with protocol heterogeneity becomes increasingly important nowadays. In addition to the IPv4/IPv6 migration, this heterogeneity appears in tunneling, some architectures (e.g., The Pseudo-Wire architecture~\cite{RFC3985} allows the emulation \--- and thus the encapsulation \--- of lower layer protocols over Packet-Switched Networks), hybrid networks (e.g., National Research and Education Networks \--- NRENs \--- which may have optical and IP interconnection points), and last but not least, most carrier-grade networks, which have separate control planes for IP and Transport layers. In all these contexts, a unified control plane would be very useful for optimizing the network resources and reduce operational and management costs. 

OpenFlow is a chance to design such a control plane. Some previous works~\cite{das2010packet,liu2013field} present an OpenFlow-based architecture to achieve this challenge, but they only focus on the convergence of packet and circuit networks. Other works tackle the traffic engineering problem in SDNs but circumscribe it on a single layer~\cite{Agarwal2013} or in the IPv4/IPv6 migration context~\cite{Li2014}. However, an important problem to solve remains the path computation process in a multi-layer context. Taking into account the adaptation functions is not trivial and classical algorithms such as Dijkstra's one \cite{Dijkstra1959} cannot achieve the task as they do not handle these functions.

Here, we design several algorithms to compute shortest paths dealing with protocol changes and adaptation functions.
\subsection*{Our contributions:}
\begin{enumerate}
\item We widely generalize the model and the polynomial algorithms described in Lamali~\textit{et~al.}~\cite{Lamali2012, Lamali2013} to perform path computation in multi-layer networks (without bandwidth constraint). Our model takes into account all possible types of protocol changes (encapsulation, conversion, etc.) and any additive metric. We drastically improve the algorithm complexity and realize the first implementation, showing their efficiency on two real topologies.
\item For simulation purposes, we empirically study the distribution of adaptation functions over the network nodes and its impact on feasible path existence. We exhibit a phase transition phenomenon, i.e., a gap where the probability of existence of a feasible path hugely increases. 
\item We prove that path computation in multi-layer networks under bandwidth constraint is $\mathsf{NP}$-complete even with two protocols and on symmetric graphs, thus improving a result of Kuipers and  Dijkstra~\cite{K09}. We also obtain results on the complexity of some subproblems: It is polynomial on \emph{Directed Acyclic Graphs} (DAG) and the general problem is not approximable. We propose a new heuristic to resolve the problem and show its efficiency through simulations. 
\item We propose the first algorithm to perform path computation in multi-layer networks under several QoS constraints by adapting the Self-Adaptive Multiple Constraints Routing Algorithm (SAMCRA \--- Van~Mieghem and Kuipers~\cite{Van04}) to the multi-layer context. We study its scalability through simulations.
\end{enumerate}

The paper is organized as follows: Section~\ref{sec:prob_state} describes the problem of path computation in  multi-layer networks and recalls the related work; Section~\ref{sec:model} formalizes the problem and describes our model of multi-layer network; Section~\ref{sec:without} proposes algorithms to perform path computation without bandwidth constraint and shows their efficiency through simulations, it also studies the phase transition phenomenon in multi-layer networks; Section~\ref{sec:with} studies the complexity  of path computation under bandwidth constraint and proposes heuristic solutions to tackle the problem; Section~\ref{sec:qos2} proposes the first algorithm computing paths under additive QoS constraints and studies its scalability; finally, Section~\ref{sec:conclusion} concludes the paper.

\section{Path computation in Multi-Layer Networks}
\label{sec:prob_state}
\subsection{Connectivity in multi-layer networks}
We aim to present the different concepts of path computation in multi-layer networks through an example. While this example relates to multi-domain multi-layer networks, the underlying problem of path computation is the same as in a single domain network\footnote{The algorithms presented in this paper can be applied in a single-domain or a multi-domain context. For the latter, however, a mechanism for sharing the network information (such as the topology) is needed. This can be done through a PCE for example~\cite{PCE}.}. Figure~\ref{fig:realnet} (inspired by the Inter-Provider Reference Model~\cite{RFC5659}) depicts a network involving multiple domains and  adaptation function capabilities of network equipments: A company owning a Local Area Network (LAN) wishes the Virtual Machines (VMs) of a data-center to be within the same routing domain (for instance through a Layer 2 VPN or a Generic Routing Encapsulation tunnel). Hence, the switches of the LAN and the VMs of the data-center must communicate through Ethernet datagrams and a path has to be determined across the Domains~$1$ and $2$.
\begin{figure*}[htb]
\centering
\includegraphics[width=\textwidth]{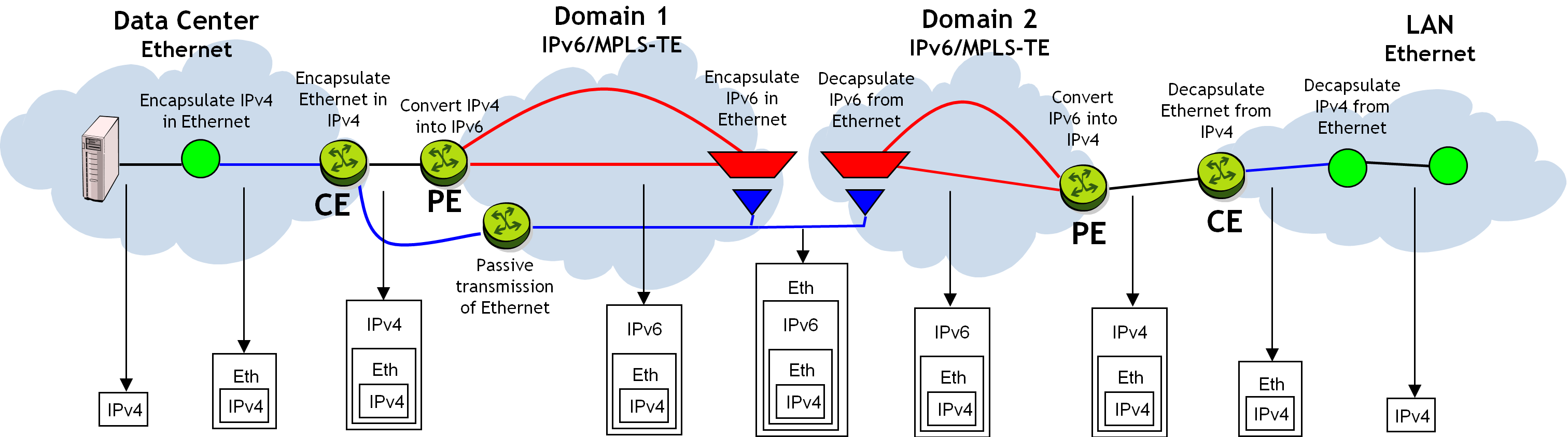}
\caption{Carrier-grade network comprising several domains and different layers.}
\label{fig:realnet}
\end{figure*}

On Figure~\ref{fig:realnet}, Domains~$1$ and $2$ use IPv6/MPLS-TE technology and are linked by equipments providing Ethernet encapsulation and decapsulation. The Provider Edge (PE) of Domain~$1$ is linked to the Customer Edge (CE) of the data-center. The adaptation capabilities of each node are shown above it.
An example of feasible path would cross the PE of Domain~$1$ converting IPv4 packets into IPv6 ones. Then it would apply the encapsulation and decapsulation of the border routers of Domains~$1$ and $2$ respectively and the PE of Domain~$2$ would apply a conversion of IPv6 packets into IPv4 ones. The protocol stacks of the packets at each stage are illustrated at the bottom of Figure~\ref{fig:realnet}. As an example of unfeasible path, a direct Ethernet connection between the CE of the data-center and the border router of Domain~$1$ appears. This configuration leads to a decapsulation of an IPv6 packet from an Ethernet datagram (by the border router of Domain~$2$) whereas at this stage the datagram encapsulates IPv4 packets.

This example depicts the constraints to comply with when computing a multi-layer (and multi-domain, in this case) path: Being physically linked is not sufficient to establish connectivity. Protocol continuity (by analogy with wavelength continuity in optical networks) must hold and the adaptation functions should occur in the right order. Moreover, feasible paths can involve loops and their subpaths are not necessary feasible~\cite{K09,Dijkstra2009}. Nowadays, such paths are manually determined and configurations are operated and applied by scripts.
\subsection{Related work}
The initial works dealing with protocol and technology heterogeneity circumscribed the problem at the optical layer. For instance, Chlamtac~\textit{et~al.}~\cite{Chlamtac1996} described a model and algorithms to compute a path under wavelength continuity constraints. Zhu~\textit{et~al.}~\cite{Zhu2003} addressed the same problem in WDM mesh networks tackling traffic grooming issues. In~\cite{Gong2008}, Gong and Jabbari provided an algorithm to compute an optimal path under constraints on several layers: wavelength continuity, label continuity, etc. 

However, the models of these past works are not adapted to the problem of nested encapsulation and decapsulation capabilities for which a kind of stack mechanism is needed. In~\cite{Dijkstra2008}, Dijkstra~\textit{et~al.} addressed this issue  in the context of the ITU-T G.805 recommendations on adaptation functions. They stressed the lack of solutions on path computation. Kuipers and Dijkstra~\cite{K09} demonstrated that the problem of path computation with encapsulation and decapsulation capabilities is $\mathsf{NP}$-complete under bandwidth constraint. They proposed a Breadth-First Search (BFS) algorithm that explores all possible paths until finding a feasible one. In~\cite{Lamali2012,Lamali2013}, Lamali~\textit{et~al.} demonstrated that the problem is polynomial if the bandwidth constraint is relaxed. Their approach was to model the network as a Push-Down Automaton and to use automata and language theory tools to compute a shortest feasible path, but only considering the number of hops or adaptation functions. More recently, Iqbal~\textit{et~al.}~\cite{Iqbal2015} underlined the need of path computation algorithms in NRENs. They proposed a new matrix-based model for multi-layer networks and algorithms based on $k$-shortest paths and \textit{LOOK-AHEAD} methods. However, the model deals with technologies\footnote{A \textit{technology} is an exhaustive description of the protocol stack at some node, e.g.,  IP over Ethernet over ATM.} instead of protocols. Thus, the nested protocols are not transparent to the nodes. Moreover, the proposed exact algorithm is exponential and can compute only loopless feasible paths.
\subsection{Proposed approach}
Our goal is to study the path computation problem in a multi-layer context and to propose efficient algorithms to resolve it. To this end, we focus on three cases: Path computation without bandwidth constraint (by adapting the language theoretic approach of Lamali\textit{~et~al.}~\cite{Lamali2013}), under bandwidth constraint (by using graph transformation in order to overcome the problem complexity) and under several QoS constraints. The simulations showing the efficiency of our algorithms follow a methodology based on the probabilistic distribution of the adaptation functions over the nodes.
\section{Model and problem formalization}
\label{sec:model}
This section describes a mathematical model of multi-layer networks and formalizes the notion of path feasibility. 
\subsection{Multi-layer network model}
\noindent{\bf Notation convention.}  In order to avoid confusion, lowercase letters denote protocols (e.g., $a,b,c,x,y$) or functions (e.g., $f,h,\ell$). Capital letters denote nodes and links (e.g., $U,V,E$). Finally, calligraphic letters denote sets (e.g., $\Gg,\V,\E$).
\label{sec:graph_model}

We consider a multi-layer network as a $4$-tuple $\N=(\Gg, \alphabet, \AD, h)$ where:
\begin{itemize}
\item $\Gg=(\V,\E)$ is a directed graph modeling the network topology. The set of nodes $\V$ models the routers of the network. The set of edges $\E$ models the physical links between the routers. 
\item $\alphabet=\{a,b,c, \dots \}$ is the set of protocols available in the network, but not necessarily at each router.
\item For each node $U\in\nodes$, $\AD(U)$ is the set of adaptation functions available on node $U$. These functions are: 
\begin{itemize}
\item \textit{Conversion}: A protocol $a$ is converted into a protocol $b$ without any change of the possible underlying protocols. This function is denoted by $(a\rightarrow b)$. E.g., Wavelength conversion on the optical layer, IPv4 to IPv6, etc.
\item \textit{Passive function}: A protocol $a$ is left as it is. It is a classical retransmission without any protocol change and can be considered as a special case of protocol conversion where $a=b$. Thus it is denoted by $(a\rightarrow a)$.
\item \textit{Encapsulation}: A protocol $a$ is encapsulated in a protocol $b$. It is denoted by $(a\rightarrow ab)$. 
\item \textit{Decapsulation}: A protocol $a$ is decapsulated from a protocol $b$. It is denoted by $\overline{(a\rightarrow ab)}$.
\end{itemize}
\item $h : \V\times\AD\times \V \rightarrow \Re_+$ is the \textit{weight} function. The value $h(U,f,V)$ (where $U,V\in \V$ and $f\in\AD(U)$ ) is the cost of using the link $(U,V)$ with the adaptation function $f$ on $U$. Hence, function $h$ allows representing any additive metric either associated only to the links or to both links and adaptation functions.
\end{itemize}

\subsection{Path feasibility}
\label{sec:path_feas}
Let $(S,D)$ be a pair of nodes in $\Gg$ corresponding to the source and the destination of the path to be computed.
We consider a path from $S$ to $D$ as a sequence of nodes and adaptation functions $Sf_0U_1f_1U_2f_2\dots U_nf_nD$ where each $U_i$, $i=1,\ldots,n$, is a node and each $f_i$ is an adaptation function ($f_0$ being fictitious). A path is \textit{feasible} if: 
\begin{enumerate}
\item The sequence $SU_1U_2\dots U_nD$ is a path in $\Gg=(\nodes,\E)$ and each $f_i \in \AD(U_i)$;
\item Each encapsulated protocol is decapsulated before reaching $D$ according to its encapsulation order and \textit{protocol continuity} must hold (i.e., if the sequence contains a function $f_i$ s.t. $f_i=(a \rightarrow b)$, $a,b \in \alphabet$, then $f_{i+1}=(b \rightarrow a')$ or $f_{i+1}=(b \rightarrow b a')$ or $f_{i+1}=\overline{(a'\rightarrow a'b)}$, $a'\in \alphabet$).
\end{enumerate}
Actually, the protocol sequences of feasible paths can be characterized as a well-parenthesized language~\cite{Lamali2013}.

\section{Path computation without bandwidth constraint}
\label{sec:without}
This section proposes a polynomial algorithm to resolve the path computation problem without bandwidth constraint and evaluates it through simulations.
\subsection{A polynomial algorithm for path computation}
Lamali~\textit{et~al}.~\cite{Lamali2013} proposed a language theoretic approach to compute a shortest feasible path (involving encapsulations and decapsulations of protocols) in a multi-layer network. The metric considered was the number of hops or of encapsulations in the path. The approach comprises the following steps:
\begin{enumerate}
\item Consider the set of protocols as an alphabet and convert the multi-layer network into a Push-Down Automaton (PDA);
\item If the considered metric is the number of encapsulations, transform the automaton in order to bypass passive transitions;
\item Convert the PDA to a Context-Free Grammar (CFG);
\item \label{it:shortest}Compute the shortest word generated by the CFG. It is the protocol sequence of a shortest path;
\item \label{it:path}Compute a shortest path from this sequence.
\end{enumerate}
We made several improvements to these algorithms:

\begin{itemize}
\item The PDA building is modified in order to support protocol conversion by adding a new transition type;
\item The PDA transitions are weighted in order to reflect the weight function. Thus, our algorithm computes the shortest path according to any additive metric (instead of just the number of hops or encapsulations); 
\item The PDA transformation is no longer useful thanks to the weight function: Simply put $h(U,f,V)=1$ (where $U,V\in \V$ and $f\in\AD(U)$) for all triples where $f$ is an encapsulation, and $h(U,f,V)=0$ for all other triples. It is also possible to set different weights to each type of encapsulation and minimize the path cost according to these weights;
\item The conversion of the PDA into a CFG is adapted: As in~\cite{Lamali2013}, each transition from the PDA is converted into a production rule set in the CFG according to a method described in~\cite{Hop06}. However, the transition weights are assigned to the corresponding production rules;
\item Step~\ref{it:shortest} is different: Since the production rules are weighted, the goal is no longer to compute the shortest word but the word having the minimum weight derivation tree. This is done thanks to Knuth's algorithm described in~\cite{Knuth77}. This word corresponds to the protocol sequence of a shortest path to compute;
\item The algorithm computing the path matching the protocol sequence is modified in order to take into account the weights.
\end{itemize}
Due to the lack of space, we cannot detail our improved algorithm. The interested reader can find it (together with its correctness proof and complexity study) in Appendix~\ref{appendix:algos}.

Additionally to these improvements, the algorithm complexity is drastically decreased. In~\cite{Lamali2013},  Step~\ref{it:shortest} has a complexity of $O(|\alphabet|^8\times|\V|^7)$ in the worst case, which is the highest complexity in the whole process.
Implementing Knuth's algorithm with Fibonacci heaps gives an $O(|\mathcal{Q}|\log|\mathcal{Q}|+\mathcal{|R|})$ complexity, where $|\mathcal{Q}|$ is the number of nonterminals in the CFG and $|\mathcal{R}|$ is the number of production rules~\cite{Tarjan1987}. Since $|\mathcal{Q}|=O(|\alphabet|^3\times|\V|^2)$ and $|\mathcal{R}|=O(|\alphabet|^5\times|\V|^2\times|\E|)$ (see Appendix~\ref{appendix:algos}), the complexity of the whole process is:
\[O\left(|\alphabet|^5\times|V|^2\times|\E|\right) \] This is a significant improvement compared to  the complexity  $O(|\alphabet|^8\times|\V|^7)$ in~\cite{Lamali2013}.

\subsection{Simulations}
\label{sec:simul_without}
We implemented our algorithm (called PDA) and compared it to a classical BFS approach.

\subsubsection{Networks used for the simulations and methodology}
Large multi-layer topologies are generally not available. Some public ones as the Internet2 network~\cite{Internet2} are not large enough to show the scaling of our algorithm. Thus we performed simulations on two topologies described in~\cite{Mahajan2002}:
\begin{itemize}
\item Topology~$T1$ is a simplified version of Time Warner network. It has $41$ nodes and $296$ directed links.
\item Topology~$T2$ corresponds to the network of Exodus as in 2002. It has $79$ nodes and $294$ directed links.
\end{itemize}

Since these topologies are not layered, the adaptation functions are randomly allocated to the nodes. For an alphabet~$\alphabet$, there are $3|\alphabet|^2$ possible adaptation functions (for each ordered pair of protocols: a conversion, an encapsulation and a decapsulation). For each node $U$, each of these adaptation functions is available on $U$ with probability $p$. The source and the destination nodes are the diameter extremities, which corresponds to $5$ (resp. $10$) hops for Topology~$T1$ (resp. $T2$).

\subsubsection{Phase transition in path feasibility}
\label{sec:trans_phase}
Depending on the network topology and the adaptation function distribution, there is not always a feasible path. It is interesting to know the probability of a feasible path existence according to probability $p$ in order to set appropriate parameters for the simulations. In case of path existence, knowing the probability that the shortest one involves loops allows comparing the different algorithms (some of them allow loops and others do not). To compute this probability, we performed $200$ runs for each value of $p$ and counted the number of times there was a feasible path.

Figure~\ref{fig:phase_trans} shows the evolution of feasible path existence probability according to $p$ and the proportion of shortest paths that involve loops. Not surprisingly, the probability of feasible path existence grows according to $p$. On both topologies, the probability of path existence reaches $50\%$ when $p=0.22$ and follows a phase transition phenomenon. For example, in the interval $p\in[0.10, 0.38]$, the probability of path existence in Topology~$T1$ grows from $5\%$ to $90\%$. This interval is the most suitable to perform simulations. The phase transition phenomenon also holds with more than $2$ protocols. The more the number of protocols is high, the more the phase transition is shifted to the left.
If there are few feasible paths (for small $p$), the probability that the shortest ones involve loops is high. However, this probability quickly decreases. For example, for $p>18\%$, the proportion of shortest paths involving loops is less than $20\%$ in Topology~$T1$. The trend of this proportion is not clear in $T2$, however it is less than $21\%$ if $p>0.22$.
\begin{figure} 
\centering
\includegraphics[width=0.45\textwidth]{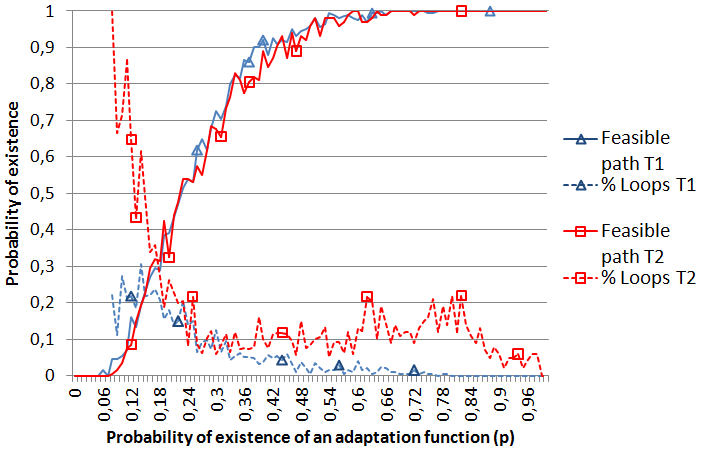}
\caption{Probability of existence of a feasible path (and a loop in the shortest one) according to the probability of existence of an adaptation function.}
\label{fig:phase_trans}
\end{figure}
The phase transition phenomenon can be seen in~\cite{Iqbal2015}. But the results consider only loopless paths and the distribution deals with technologies rather than adaptation functions. 

\subsubsection{Simulation results}
\label{sec:simul_without_bandwidth}
Our algorithm is compared to a classical BFS which explores all possible paths until reaching the destination. During the exploration process, all \textit{dominated}\footnote{In this context, a path dominates another one if they have the same extremities and the same protocol stack, and the first path is shorter.} paths are deleted. BFS can be seen as a version of the algorithm in~\cite{K09} where the bandwidth constraint is relaxed. The first results showed that BFS algorithm is extremely slow even for small values of $p$ (processing time of the order of several hours). It was impossible to perform a comparison with our algorithm. Due to this tremendous running time, we fixed a maximum length to the explored paths by BFS algorithm. If a path exceeds $10$ hops (resp. $14$ hops) on Topology~$T1$ (resp. $T2$), it is deleted and no more considered. We performed $100$ runs for each value of $p$ and averaged the processing time.
\begin{figure} 
\centering
\includegraphics[width=0.45\textwidth]{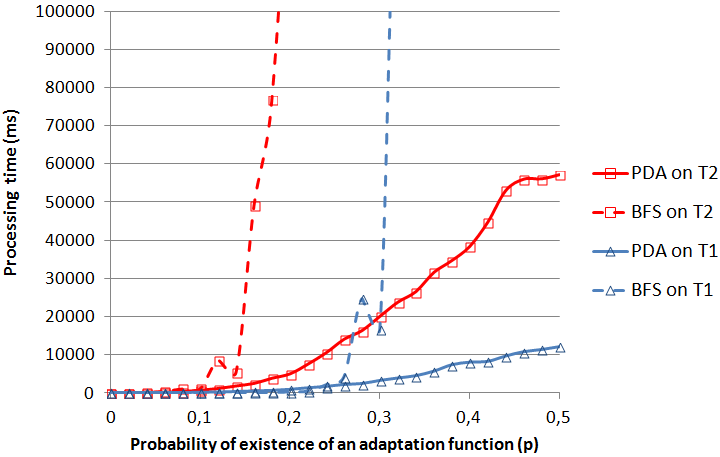}
\caption{Comparison of processing time of PDA algorithm and BFS on Topologies $T1$ and $T2$.}
\label{fig:pda_vs_bfs}
\end{figure}
Figure~\ref{fig:pda_vs_bfs} shows the processing time of PDA algorithm and BFS algorithm on Topologies $T1$ and $T2$ according to the values of $p$. For small values of $p$ ($<0.22$ for $T1$ and $<0.04$ for $T2$) BFS algorithm is faster than PDA. However, the processing time of BFS explodes. We cannot put it on Figure~\ref{fig:pda_vs_bfs} because it would be unreadable. For example, the processing time of BFS algorithm on Topology~$T2$ for $p=0.24$ is more than $14$ minutes, while that of PDA algorithm is $10$ seconds. On Topology~$T1$, for $p=0.38$, the processing time of BFS algorithm is more than $7$ minutes, while that of PDA algorithm is $7$ seconds.  These results show that our algorithm clearly outperforms the BFS approach.
\section{Addressing bandwidth constraint}
\label{sec:with}
This section studies the complexity of path computation under bandwidth constraint and proposes heuristic solutions to resolve the problem.
\subsection{Problem formalization}
For Traffic Engineering purposes, a feasible path may be constrained by a minimal bandwidth. But it is possible that feasible paths in a multi-layer network involve loops (i.e., involving the same link several times but using different protocols). It implies that the bandwidth constraint is no longer prunable: Even if the links with not enough bandwidth are deleted by topology filtering prior to path computation, other links can have enough bandwidth if they are selected once but not if more. For example, if a link has a capacity of $10$Gbps and the bandwidth constraint is $5$Gbps, then this link cannot be crossed more than twice. The (optimization) problem of computing the shortest path in a multi-layer network under bandwidth constraint is defined as follows:
\begin{equation}
\label{def:BFP}
\begin{split}
\min &\ h(\P)=\sum_{(U,f,V)\in\P}h(U,f,V)\\
 s.t.& \left\lbrace\begin{array}{l}
\P \text{ is a feasible path between $S$ and $D$}\\
 \\
\min_{E\in \P}\dfrac{\qos_b(E)}{\nb(E)}\geq \qos_b^{min} \\
\end{array} \right.
\end{split}
\end{equation}
where $nb(E)$ is the number of times a link $E$ is crossed by path $\P$, $\qos_b(E)$ is the bandwidth capacity of $E$ and $\qos_b^{min}$ is the bandwidth constraint.
  
\subsection{Path computation complexity under bandwidth constraint}
The bandwidth constraint impacts the complexity of feasible path computation. In a single-layer network, computing a path under bandwidth constraint is trivial: It suffices to prune all the links without enough bandwidth. This is no longer possible in a multi-layer network. In fact, the decision problem is $\mathsf{NP}$-complete as shown by Kuipers and Dijkstra~\cite{K09}. But this proof does not work on symmetric directed graphs\footnote{A symmetric directed graph is a graph where a link $(U,V)$ exists if and only if the reverse link $(V,U)$ exists.}. However, most communication networks are symmetric. We show that the decision version of the problem remains $\mathsf{NP}$-complete even with two protocols and in a symmetric graph. Consider the following problem:

{\bf Problem (\ref{def:BFP}')}. Given a multi-layer network $\N=(\Gg=(\V,\E), \alphabet, \AD,h)$, a function  assigning to each link $E\in\E$ an available bandwidth $\qos_b(E)$, a bandwidth constraint $\qos_b^{min}$ and a pair $S$ and $D$ of nodes in $\V$. Is there a feasible path from $S$ to $D$ satisfying the bandwidth constraint?

\begin{proposition}
\label{prop:NPsym}
Problem~(\ref{def:BFP}') is $\mathsf{NP}$-complete with two protocols even if $\Gg=(\V,\E)$ is a symmetric directed graph.
\end{proposition}

\begin{proof}
Clearly, the problem is in $\mathsf{NP}$. Thus, we only detail the proof of $\mathsf{NP}$-hardness.

First consider the problem of finding a Hamiltonian path in a symmetric directed graph between two nodes $S'$ and $D'$. Call this problem SYM-HAM. SYM-HAM is $\mathsf{NP}$-complete (for a detailed proof, see Appendix~\ref{app:np-complete}).

Now we provide a polynomial reduction from SYM-HAM to Problem~(\ref{def:BFP}') restricted to symmetric directed graph and two protocols. Given an instance of SYM-HAM, i.e., a symmetric directed graph $\Hh=(\V',\E')$ and a pair of nodes $(S',D')$, we build an instance of Problem~(\ref{def:BFP}'), i.e., a network $\N=(\Gg, \alphabet, \AD, h)$ and a pair of nodes $(S,D)$ as following:

\noindent{\bf Step~1: Splitting the nodes.} For each node $U'\in \V'$, four nodes $U_1,U_2, U_3$ and $U_4$ are created in $\Gg$. Links $(U_i,U_{i+1})$ and $(U_{i+1},U_{i})$ are created for $i=1,\dots,3$. For each link $(U',V')\in \E'$, a link $(U_1,V_1)$ is created in $\Gg$. This step is illustrated on Figure~\ref{fig:reduc_nodes_sym}.

\noindent{\bf Step~2: Adding a tail.} $\Gg=(\V,\E)$ is augmented by a set $\C=\{C_0,\dots, C_{n+1}\}$ of nodes ($n=|\V'|$), where $C_0=S$ is the source node. There are a link $(C_i,C_{i+1})$ and a link $(C_{i+1},C_i)$ for $i=0,\dots,n$. Moreover, there is also a link from $C_{n+1}$ to $S_1$ (the first node resulting from the splitting of $S'$) and conversely. Figure~\ref{fig:reduc_graph_sym} shows this construction. Finally, two nodes $X$ and $D$ are added, as well as the four links $(D_1,X),(X,D_1),(X,D)$ and $(D,X)$ (recall that $D_1$ is the first node resulting from the splitting of $D'$, see Step~$1$).

\noindent{\bf Step~3: Allocating the adaptation functions and available bandwidth.} All the links have available bandwidth $1$. The bandwidth constraint is set to $1$. Thus, any feasible path must cross a link at most once. There is no possible loop. 
Let the set of protocols be $\alphabet=\{a,b\}$. Node $S$ emits packets of protocol $a$. For $i=1\dots,n$, each node $C_i$ in the tail can encapsulate protocol $a$ in itself. Node $C_{n+1}$ can only encapsulate $a$ in $b$. For each node $U'\in\V'$, node $U_1$ can encapsulate any protocol in $b$. Node $U_2$ can either decapsulate protocol $b$ from itself or passively transmit protocol $a$. Node $U_3$ can either decapsulate protocol $a$ from $b$ or passively transmit protocol $a$. Node $U_4$ is able to decapsulate protocol $a$ from itself. Finally, node $X$ can decapsulate protocol $a$ from $b$. Table~\ref{tab:functions} summarizes the allocation of the adaptation functions.

\begin{figure*}[htb]
\centering
\includegraphics[width=0.65\textwidth]{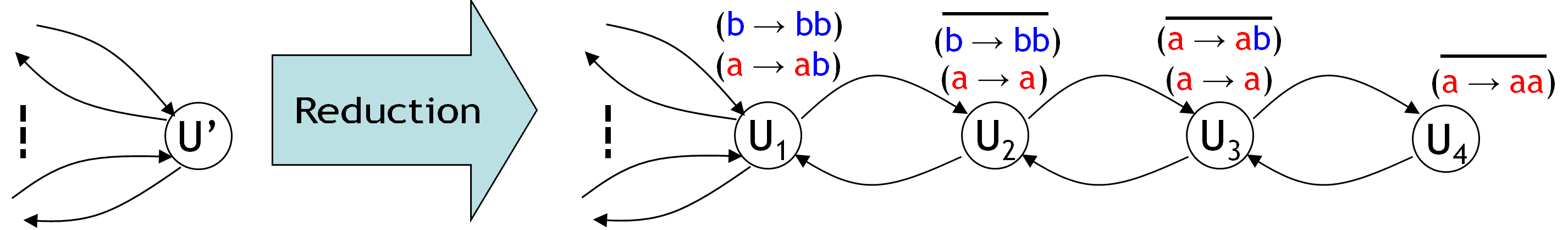}
\centering
\caption{Reduction from SYM-HAM to feasible path under bandwidth constraint (node splitting).}
\label{fig:reduc_nodes_sym}
\end{figure*}

\begin{figure*}[htb]
\centering
\hspace{10mm}
\includegraphics[width=0.98\textwidth]{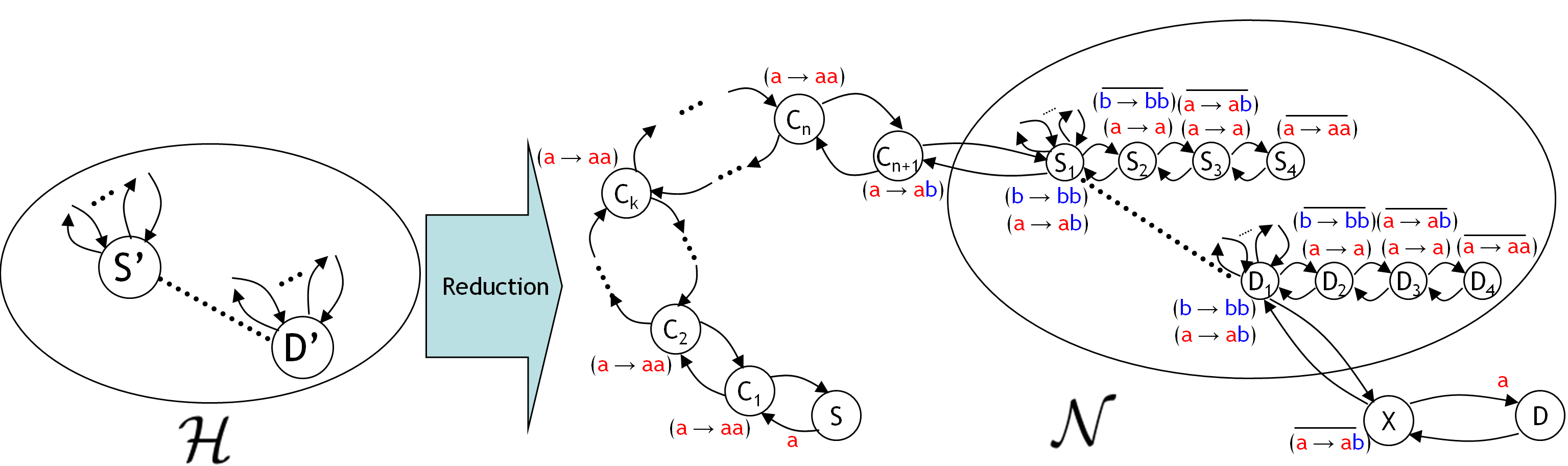}
\caption{Reduction from SYM-HAM to feasible path under bandwidth constraint (graph transformation).}
\label{fig:reduc_graph_sym}
\end{figure*}

\begin{table}
	\centering
	\renewcommand{\arraystretch}{1.7}
		\begin{tabular}{|l|l|}
		\hline
			Node & Adaptation functions \\
			\hline
			$C_i,\ i=1\dots n$  & $(a\rightarrow aa)$ \\
			\hline
			$C_{n+1}$  & $(a\rightarrow ab)$ \\
			\hline
			$U_1\ s.t.\ U'\in\V'$  & $(b\rightarrow bb)$, $(a\rightarrow ab)$ \\
			\hline			
			$U_2\ s.t.\ U'\in\V'$  & $\overline{(b\rightarrow bb)}$, $(a\rightarrow a)$ \\
			\hline			
			$U_3\ s.t.\ U'\in\V'$  & $\overline{(a\rightarrow ab)}$, $(a\rightarrow a)$ \\
			\hline			
			$U_4\ s.t.\ U'\in\V'$  & $\overline{(a\rightarrow aa)}$ \\
			\hline			
			$X$  & $\overline{(a\rightarrow ab)}$ \\
			\hline				
		\end{tabular}
	\caption{The adaptation functions available on the nodes in the polynomial reduction.}
	\label{tab:functions}
\end{table}

Now, we prove that there is a Hamiltonian path from $S'$ to $D'$ in $\Hh$ if and only if there is a feasible path from $S$ to $D$ in $\N$ that satisfies the bandwidth constraint.
First, assuming that there is a Hamiltonian path from $S'$ to $D'$ in $\Hh$, we construct a feasible path $\P$ in $\N$ as follows: Starting from $S$ in $\N$, $\P$ crosses the tail and each $C_i$ ($i=1\dots n$) adds an occurrence of protocol $a$ in the stack of encapsulated protocols. Then crossing $C_{n+1}$ adds $b$ as current protocol. Thus, at the end of the tail, there are $n+1$ encapsulated protocols $a$ (the one emitted by $S$ and $n$ occurrences added in the tail) and the current protocol is $b$. Following the same node order as in the Hamiltonian path, replace each occurrence of a node $U'\in\V'$ (including $S'$ and $D'$) in the Hamiltonian path by the sequence:
\begin{equation}
\label{eq:seq}
\begin{split}
& U_1(b\rightarrow bb) U_2 \overline{(b\rightarrow bb)} U_3 \overline{(a\rightarrow ab)} U_4 \overline{(a\rightarrow aa)} U_3 (a\rightarrow a)\\
&  U_2 (a\rightarrow a) U_1 (a\rightarrow ab)
\end{split}
\end{equation}
Thus, at node $U_1$ an encapsulation of protocol $b$ occurs, at $U_2$ protocol $b$ is decapsulated, at $U_3$ it is decapsulated again, and at $U_4$ protocol $a$ is decapsulated. Path $\P$ then crosses passively nodes $U_3$ and $U_2$, and finally encapsulates protocol $b$ at $U_1$. Thus, at each time the path crosses a Sequence~\eqref{eq:seq}, then one occurrence of protocol $a$ is removed from the protocol stack. Crossing all $U_4$ s.t. $U'\in\Hh$ removes all encapsulated occurrences of protocol $a$ except the first one. When the path  leaves $D_1$ to reach node $X$, the current protocol is $b$ and there is a last occurrence of protocol $a$ which is encapsulated. Finally, node $X$ decapsulates protocol $a$ from protocol $b$ and node $D$ receives protocol $a$ as emitted by $S$. Thus, $\P$ is a feasible path, and each link is crossed at most once, the bandwidth constraint is satisfied. 

Conversely, we show that from any feasible path $\P$ satisfying the bandwidth constraint in $\N$, one can extract a Hamiltonian path between $S'$ and $D'$ in $\Hh$. A feasible path must cross all nodes $U_4$ s.t. $U'\in\V'$ in order to decapsulate all occurrences of protocol $a$ encapsulated when crossing the tail. Thus, it involves Sequence~\eqref{eq:seq} for all $U'\in \V'$. By removing the tail part and the nodes $X$ and $D$ from $\P$ and replacing each occurrence of Sequence~\eqref{eq:seq} by the corresponding node $U'$, the resulting path starts from $S'$ and crosses all the nodes in $\Hh$ before reaching $D'$. The only problem is the possibility that there are other sequences than Sequence~\eqref{eq:seq} in the remaining path. There are two possible cases:
\begin{itemize}
\item An \emph{incomplete} Sequence~\eqref{eq:seq} where $U_4$ is not reached (e.g., $U_1 f U_2 f' U_3 f'' U_2 f''' U_1$): This cannot happen because such a sequence forbids to reach $U_4$ later, and thus one encapsulated occurrence of protocol $a$ is never decapsulated and $\P$ cannot be feasible. Such a sequence cannot occur after an occurrence of Sequence~\eqref{eq:seq} on the same nodes because if a node $U_i$ ($i=2,3$) is reached in a Sequence~\eqref{eq:seq} it cannot be reached again due to the bandwidth constraint.
\item A sequence $U_1 f V_1 f' W_1$: Let $\P$ be a feasible path from $S$ to $D$ containing a sequence $U_1 f V_1 f' W_1$ (where $U_1$ and $W_1$ may be the same node). These three nodes can only encapsulate protocol $a$ or $b$ in protocol $b$. Thus, after crossing such a sequence, there are three occurrences of protocol $b$ on the top of the protocol stack. However, in network $\N$, there is no possible sequence of nodes and adaptation functions able to decapsulate protocol $b$ three consecutive times. Thus, $\P$ is not feasible.
\end{itemize}
Thus, if a feasible path exists, then it contains only one occurrence of Sequence~\eqref{eq:seq} for each node $U'\in\V'$. Replacing each Sequence~\eqref{eq:seq} by the corresponding node in $\V'$ induces a Hamiltonian path in $\Hh$. This concludes the proof. 
\end{proof}

Unfortunately, the previous negative result implies:
\begin{corollary}
Problem~\eqref{def:BFP} is not approximable (unless $\mathsf{P}=\mathsf{NP}$).
\end{corollary}
	
\begin{proof}
Since the existence of a feasible path (independently of its cost) is $\mathsf{NP}$-complete to decide, any polynomial approximation algorithm would imply $\mathsf{P}=\mathsf{NP}$.
\end{proof}	
On the other hand, the problem is tractable on some particular topologies:
\begin{corollary}
\label{cor:DAG}
Problems~\eqref{def:BFP} and (\ref{def:BFP}') are polynomial if the graph $\Gg=(\V,\E)$ is a Directed Acyclic Graph (DAG).
\end{corollary}

\begin{proof}
The $\mathsf{NP}$-completeness of Problem~(\ref{def:BFP}') results from the fact that the bandwidth constraint is not prunable when feasible paths involve loops. In a DAG, every link is involved at most once in a feasible path due to the absence of cycles. Thus the bandwidth constraint is prunable and the problem can be resolved using the method described in Section~\ref{sec:without}.
\end{proof}

\subsection{DAG Heuristic}
As seen in Section~\ref{sec:trans_phase}, shortest feasible paths involving loops are infrequent (for $p>20\%$). Combining this fact with Corollary~\ref{cor:DAG} suggests a heuristic to compute feasible path under bandwidth constraint: Convert the network into a DAG and perform the PDA algorithm to compute a shortest feasible path.

\noindent{\bf DAG Conversion.} The network is converted into a DAG in the following way:
\begin{enumerate}
\item Set the number $0$ to node $S$ and $|\V|-1$ to node $D$ (recall that $S$ and $D$ are the extremities of the graph diameter);
\item Perform a BFS algorithm starting from node $S$ and number the nodes in the visit order. The nodes at the same distance from $S$ are visited randomly, thus performing several times this heuristic does not always give the same node numbering and the same DAG;
\item Delete all the links that start at a node and end at a node with a smaller number.
\end{enumerate}

The DAG heuristic is as follows:
\begin{enumerate}
\item Convert the network into a DAG; 
\item Prune the links without enough bandwidth;
\item Perform the PDA algorithm of Section~\ref{sec:without}.
\end{enumerate}

\subsection{Simulations}
We study the efficiency of the DAG heuristic (called DAG-PDA) and compare it with the algorithm of Kuipers and Dijkstra~\cite{K09}. The latter is an exact (and thus exponential) algorithm that performs a BFS and explores all the paths that are not dominated and that satisfy the bandwidth constraint. As in Section~\ref{sec:simul_without_bandwidth}, the BFS algorithm is slow. Thus, we also compare our algorithm with DAG-BFS algorithm, where the network is converted into a DAG before performing the BFS.
The simulation conditions (parameters, topology, number of runs, etc.) are the same as in Section~\ref{sec:simul_without_bandwidth}. The bandwidth capacity of the links is randomly and uniformly selected in the set $\{1,2,\dots,10\}$. The bandwidth constraint is set to $2$.
\subsubsection{Comparison of the feasibility ratio}
Converting the network topology into a DAG deletes some feasible paths in the original network. We measure how much feasible paths are lost by comparing the probability of feasible path existence before and after the DAG conversion according to the probability of existence of adaptation functions ($p$).
 \begin{figure} 
 \centering
 \includegraphics[width=0.45\textwidth]{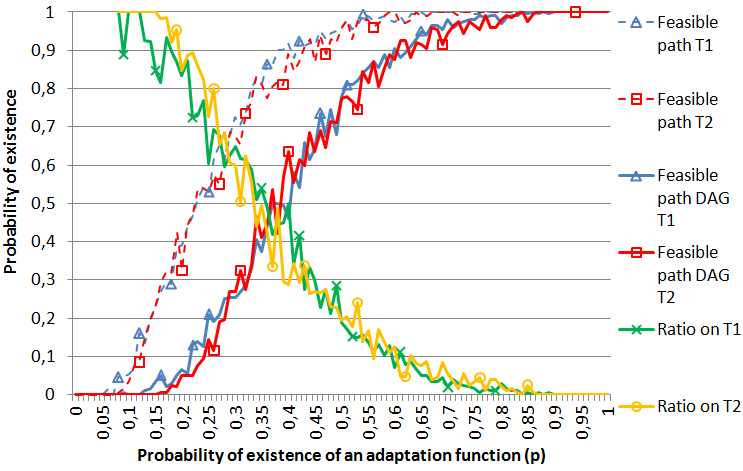}
 \caption{Probability of feasible path existence before and after DAG conversion on Topologies $T1$ and $T2$.}
 \label{fig:trans_phase_dag}
 \end{figure}
Figure~\ref{fig:trans_phase_dag} shows that the probability of feasible path existence is shifted to the right after the DAG conversion. The ratio $\frac{\text{Probability of feasible path existence in T}i}{\text{Probability of feasible path existence in DAG T}i}$ ($i=1,2$) is clearly decreasing and is less than $50\%$ if $p>0.34$, which is important but balanced by the improvement of the processing time.

\subsubsection{Comparison of the processing time}

 \begin{figure} 
 \centering
 \includegraphics[width=0.45\textwidth]{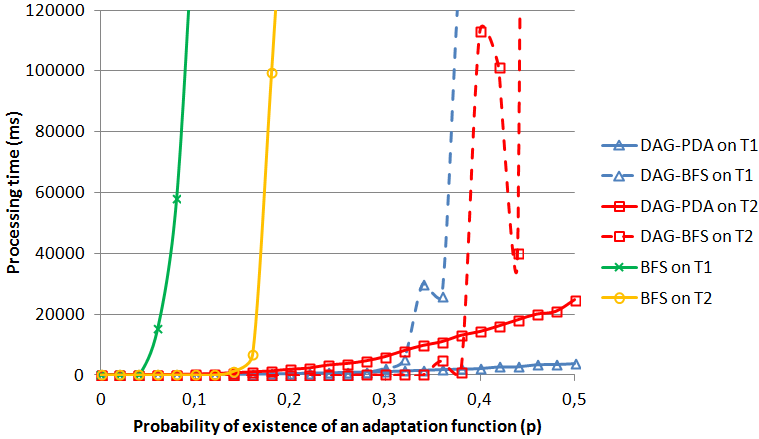}
 \caption{Comparison of the processing time of DAG-PDA, DAG-BFS and BFS algorithms on Topologies $T1$ and $T2$.}
 \label{fig:dag_pda}
 \end{figure}

Figure~\ref{fig:dag_pda} shows the processing time of DAG-PDA, DAG-BFS and BFS algorithms on both topologies according to the probability of existence of an adaptation function. BFS algorithm is slow even for small values of $p$.
For $p<0.3$ (resp. $0.4$) on Topology~$T1$ (resp. $T2$), DAG-BFS is faster than DAG-PDA. Beyond these values, the processing time of DAG-BFS explodes. For example, for $p=0.5$, the processing time of DAG-BFS is more than $35$ minutes on Topology~$T1$ and more than $53$ minutes on Topology~$T2$, while that of DAG-PDA is $3.8$ seconds on $T1$ and $24$ seconds on $T2$. These results show that the DAG-PDA algorithm is clearly faster when there is a significant number of adaptation functions, but the exponential DAG-BFS algorithm is faster if there are few of them (for small values of $p$).

\section{Path computation under QoS constraints}
\label{sec:qos2}
\subsection{Multi-constrained feasible path}
\label{sec:multi_cons}
Let $\N$ be a multi-layer network. Each link $E=(U,V)$ is associated to a set of $m$ additive QoS metrics $\qos(E)=(\qos_1(E),\dots,\qos_m(E))$ in addition to its available bandwidth $\qos_b(E)$. These additive metrics can be the delay, logarithm of the packet-loss, etc.

Let $\qos_b^{\min}$ be the bandwidth constraint and $\qos^{max}=(\qos_1^{\max},\qos_2^{\max}\dots,\qos_m^{\max})$ be a vector of QoS constraints, the problem of computing a shortest feasible path under these constraints is formalized as:
\begin{equation}
\label{def:MCFP}
\begin{split}
\min &\ h(\P)=\sum_{(U,f,V)\in\P}h(U,f,V)\\
 s.t.& \left\lbrace\begin{array}{l}
\P \text{ is a feasible path between $S$ and $D$}\\
 \\
\min_{E\in\P}\dfrac{\qos_b(E)}{\nb(E)}\geq \qos_b^{\min} \\
 \\ 
\sum_{E\in\P}\left(\qos_i(E)\times\nb(E)\right)\leq \qos_i^{\max},\ i=1\dots m\\
\end{array} \right.
\end{split}
\end{equation}

\subsection{Complexity of multi-constrained feasible path computation}
\label{sec:complexity}

The problem of QoS multi-constrained path computation (on a single layer) is well studied. It is well-known that the decision version associated to this problem is $\mathsf{NP}$-complete, even with $2$ additive and/or multiplicative constraints~\cite{Wang96}. Van~Mieghem and Kuipers~\cite{Van04} gave an exponential time algorithm but showed that the instances that really require an exponential computation time are infrequent. 
The classical multi-constrained path problem is a particular case of Problem~\ref{def:MCFP}, corresponding to the case where there is only one protocol and passive transitions. Thus the decision version associated to Problem~\ref{def:MCFP} is also $\mathsf{NP}$-complete.

\subsection{ML-SAMCRA}
\label{sec:qos}
As computing a multi-layer path under QoS constraints is $\mathsf{NP}$-complete, any algorithm able to solve this problem is exponential in the worst case (unless $\mathsf{P}=\mathsf{NP}$). We propose to adapt the Self-Adaptive Multiple Constraints Routing Algorithm (SAMCRA) to the multi-layer context in order to compute a shortest feasible path under QoS constraints.

SAMCRA is an exact QoS routing algorithm proposed by Van Mieghem and Kuipers~\cite{Van04}. It computes the shortest path under several (additive) QoS constraints but it ignores the feasibility constraint as defined in our paper. SAMCRA has an exponential worst case complexity, but it exhibits a reasonable processing time in practice.

\subsubsection{The main concepts of SAMCRA}
The idea of SAMCRA is to maintain a path list from the source node $S$ to all other nodes until reaching the destination node $D$. It progressively removes the paths that do not comply with the QoS constraints. The main concepts of SAMCRA are:
\begin{itemize}
\item \textit{Non-linear path length:} In SAMCRA, the path length is defined as a non-linear function of the QoS parameters of each link. It reduces the solution space to scan but the algorithm can apply with any metric. Hence, it is not a strict requirement.
\item \textit{The $k$-shortest path algorithm:} The $k$-shortest path algorithm maintains the list of the paths that are not (yet) removed from the path list. 
\item \textit{Non-dominance:} A multi-constrained path $\P$ \textit{dominates} another path $\P'$ if $\forall i, \sum_{E\in \P}\qos_i(E)\leq\sum_{E\in \P'}\qos_i(E)$ (i.e., if $\P$ is better than $\P'$ for each QoS parameter). A path $\P$ is non-dominated if there is no path which dominates it. The concept of non-dominance induces a partial order over the paths. It avoids the exploration of several paths thus substantially reducing the average complexity of SAMCRA. 
\end{itemize}
The path length definition is not impacted by the multi-layer context and using a linear path length function is not forbidden. The $k$-shortest path algorithm is not impacted either. However, the concept of dominance must be redefined to meet the path feasibility constraint and to take into account possible loops.
\subsubsection{Extension of the non-dominance definition}
A multi-layer path is characterized by its nodes but also by its protocol stack at the destination node. Thus in the algorithm path list, each path should be stored with its protocol stack at its final node. A multi-layer path can involve the same link several times. Before checking if this path complies with some QoS parameters, the parameters of each link should be multiplied by the number of times this link is involved in the path. The bandwidth constraint is not prunable in multi-layer context, the new non-dominance definition should take it into account.

A path $\P$ dominates a path $\P'$ if the four following conditions are satisfied:
\begin{itemize}
\item $\min_{E\in \P}\frac{\qos_b(E)}{nb_\P(E)}\geq\min_{E\in\P'}\frac{\qos_b(E)}{nb_{\P'}(E)}$
\item $\sum_{E\in \P}\qos_i(E)\times nb_\P(E)\leq\sum_{E\in\P'}\qos_i(E)\times nb_{\P'}(E)$ \\$\forall i=1,\dots,m$
\item $\P$ and $\P'$ have the same final node;
\item $\P$ and $\P'$ have the same protocol stack at this node.
\end{itemize}
Where $nb_\P(E)$ (resp. $nb_{\P'}(E)$) is the number of times the link $E$ is involved in path $\P$ (resp. $\P'$). According to this new definition of non-dominance, ML-SAMCRA explores all the possible paths until reaching the destination node with satisfactory QoS parameters. Along the exploration, it removes all paths that are dominated or not feasible. 

\subsection{Simulations}
We know study the efficiency of ML-SAMCRA through simulations and check if it is as scalable in a multi-layer context as SAMCRA in a single layer context.
\begin{figure} 
\centering
\includegraphics[width=0.45\textwidth]{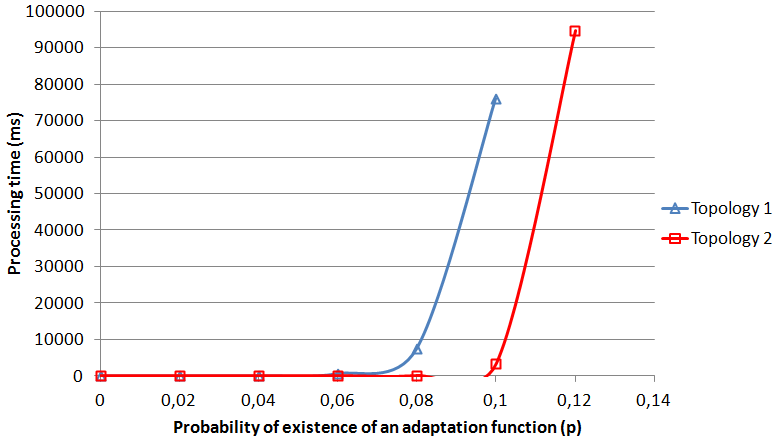}
\caption{Processing time of ML-SAMCRA Topologies $T1$ and $T2$.}
\label{fig:ml_samcra}
\end{figure}
Figure~\ref{fig:ml_samcra} shows the processing time of ML-SAMCRA on Topologies~$T1$ and $T2$ according to the probability of existence of an adaptation function ($p$). The results show that for $p>0.08$ (resp. $0.10$) on Topology~$T1$ (resp. $T2$) the processing time explodes (more than  $1$ minutes). Clearly, ML-SAMCRA does not scale above these values. There are two reasons:
\begin{enumerate}
\item The paths are less comparable in term of the new non-dominance definition: They should have the same protocol stack. As there are less dominated paths, the algorithm complexity increases;
\item Taking into account loops increases the number and the length of the paths, which also increases the algorithm complexity.
\end{enumerate}
So, path computation under QoS constraints in multi-layer networks is more complex than in single layer networks. Thus, exact algorithms are suitable only for small instances.
\section{Conclusion}
\label{sec:conclusion}

Most of carrier-grade networks manage their different layers thanks to separate control planes. Designing a unified control plane would allow the network resources to be optimized and the operational management costs to be reduced. One key problem to address is path computation taking into account the protocol heterogeneity and the multi-layer context dealing with encapsulation, conversion and decapsulation of protocols. This paper tackles this issue by partitioning it into three cases: Path computation without bandwidth constraint, under bandwidth constraint and under additive QoS constraints. For the first case, we widely generalized polynomial algorithms in the state of the art and decreased their complexity. Through simulations, we showed that they outperform previous approach in the literature. For the second case, we obtained several time complexity results and proposed efficient heuristics. Finally, we designed the first algorithm to resolve the third case. In future works, we plan to design heuristics to deal with additive QoS metrics, as the exact approach seems to be not scalable. The problem of efficient generation of random topologies being widely open, it would be interesting to analytically study the phase transition phenomenon in order to generate topologies having a suitable number of feasible paths.

\bibliographystyle{IEEEtran}
\bibliography{automata_infocom}

\begin{thebibliography}{10}
\providecommand{\url}[1]{#1}
\csname url@samestyle\endcsname
\providecommand{\newblock}{\relax}
\providecommand{\bibinfo}[2]{#2}
\providecommand{\BIBentrySTDinterwordspacing}{\spaceskip=0pt\relax}
\providecommand{\BIBentryALTinterwordstretchfactor}{4}
\providecommand{\BIBentryALTinterwordspacing}{\spaceskip=\fontdimen2\font plus
\BIBentryALTinterwordstretchfactor\fontdimen3\font minus
  \fontdimen4\font\relax}
\providecommand{\BIBforeignlanguage}[2]{{%
\expandafter\ifx\csname l@#1\endcsname\relax
\typeout{** WARNING: IEEEtran.bst: No hyphenation pattern has been}%
\typeout{** loaded for the language `#1'. Using the pattern for}%
\typeout{** the default language instead.}%
\else
\language=\csname l@#1\endcsname
\fi
#2}}
\providecommand{\BIBdecl}{\relax}
\BIBdecl

\bibitem{RFC4448}
L.~Martini, E.~Rosen, N.~El-Aawar, and G.~Heron, ``{RFC4448} - {Encapsulation
  Methods for Transport of Ethernet over MPLS Networks},'' 2008.

\bibitem{RFC6144}
F.~Baker, X.~Li, C.~Bao, and K.~Yin, ``{RFC}6144 - {Framework for IPv4/IPv6
  Translation},'' 2011.

\bibitem{RFC3985}
S.~Bryant and P.~Pate, ``{RFC3985 - Pseudo Wire Emulation Edge-to-Edge (PWE3)
  Architecture},'' 2005.

\bibitem{das2010packet}
S.~Das, G.~Parulkar, N.~McKeown, P.~Singh, D.~Getachew, and L.~Ong, ``Packet
  and circuit network convergence with openflow,'' in \emph{Optical Fiber
  Communication Conference}.\hskip 1em plus 0.5em minus 0.4em\relax Optical
  Society of America, 2010.

\bibitem{liu2013field}
L.~Liu, D.~Zhang, T.~Tsuritani, R.~Vilalta, R.~Casellas, L.~Hong, I.~Morita,
  H.~Guo, J.~Wu, R.~Mart{\'\i}nez \emph{et~al.}, ``Field trial of an
  openflow-based unified control plane for multilayer multigranularity optical
  switching networks,'' \emph{Lightwave Technology, Journal of}, vol.~31,
  no.~4, pp. 506--514, 2013.

\bibitem{Agarwal2013}
S.~Agarwal, M.~S. Kodialam, and T.~V. Lakshman, ``Traffic engineering in
  software defined networks,'' in \emph{Proceedings of the {IEEE} {INFOCOM}
  2013, Turin, Italy, April 14-19, 2013}, 2013, pp. 2211--2219.

\bibitem{Li2014}
S.~Li, Y.~Shao, S.~Ma, N.~Xue, S.~Li, D.~Hu, and Z.~Zhu, ``Flexible traffic
  engineering: When openflow meets multi-protocol ip-forwarding,'' \emph{{IEEE}
  Communications Letters}, vol.~18, no.~10, pp. 1699--1702, 2014.

\bibitem{Dijkstra1959}
E.~W. Dijkstra, ``A note on two problems in connexion with graphs.''
  \emph{Numerische Mathematik}, vol.~1, pp. 269--271, 1959.

\bibitem{Lamali2012}
M.~L. Lamali, H.~Pouyllau, and D.~Barth, ``Path computation in multi-layer
  multi-domain networks,'' in \emph{Networking (1)}, 2012, pp. 421--433.

\bibitem{Lamali2013}
------, ``Path computation in multi-layer multi-domain networks: {A} language
  theoretic approach,'' \emph{Computer Communications}, vol.~36, no.~5, pp.
  589--599, 2013.

\bibitem{K09}
F.~A. Kuipers and F.~Dijkstra, ``Path selection in multi-layer networks,''
  \emph{Computer Communications}, 2009.

\bibitem{Van04}
P.~V. Mieghem and F.~A. Kuipers, ``{Concepts of exact QoS routing
  algorithms},'' \emph{IEEE/ACM Trans. Netw.}, vol.~12, no.~5, pp. 851--864,
  2004.

\bibitem{PCE}
A.~Farrel, J.~Vasseur, and J.~Ash, ``{RFC4655 - A Path Computation Element
  (PCE)-Based Architecture},'' 2006.

\bibitem{RFC5659}
M.~Bocci and S.~Bryant, ``{RFC5659 - An Architecture for Multi-Segment
  Pseudowire Emulation Edge-to-Edge},'' 2009.

\bibitem{Dijkstra2009}
F.~Dijkstra, J.~V. der Ham, P.~Grosso, and C.~de~Laat, ``A path finding
  implementation for multi-layer networks,'' \emph{Future Generation Comp.
  Syst.}, vol.~25, no.~2, pp. 142--146, 2009.

\bibitem{Chlamtac1996}
I.~Chlamtac, A.~Farag{\'o}, and T.~Zhang, ``{Lightpath (Wavelength) Routing in
  Large WDM Networks},'' \emph{IEEE Journal on Selected Areas in
  Communications}, vol.~14, no.~5, pp. 909--913, 1996.

\bibitem{Zhu2003}
H.~Zhu, H.~Zang, K.~Zhu, and B.~Mukherjee, ``A novel generic graph model for
  traffic grooming in heterogeneous {WDM} mesh networks,'' \emph{IEEE/ACM
  Trans. Netw.}, vol.~11, no.~2, pp. 285--299, 2003.

\bibitem{Gong2008}
S.~Gong and B.~Jabbari, ``{Optimal and Efficient End-to-End Path Computation in
  Multi-Layer Networks},'' in \emph{ICC}, 2008, pp. 5767--5771.

\bibitem{Dijkstra2008}
F.~Dijkstra, B.~Andree, K.~Koymans, J.~van~der Ham, P.~Grosso, and C.~de~Laat,
  ``A multi-layer network model based on {ITU-T} {G}.805,'' \emph{Comput.
  Netw.}, 2008.

\bibitem{Iqbal2015}
F.~Iqbal, J.~van~der Ham, and F.~Kuipers, ``Technology-aware multi-domain
  multi-layer routing,'' \emph{Computer Communications}, vol.~62, pp. 85--96,
  2015.

\bibitem{Hop06}
J.~E. Hopcroft, R.~Motwani, and J.~D. Ullman, ``Introduction to automata
  theory, languages, and computation.''\hskip 1em plus 0.5em minus 0.4em\relax
  Addison-Wesley, 2006.

\bibitem{Knuth77}
D.~E. Knuth, ``{A Generalization of Dijkstra's Algorithm},'' \emph{Inf.
  Process. Lett.}, vol.~6, no.~1, pp. 1--5, 1977.

\bibitem{Tarjan1987}
M.~L. Fredman and R.~E. Tarjan, ``{Fibonacci Heaps and Their Uses in Improved
  Network Optimization Algorithms},'' \emph{J. ACM}, vol.~34, no.~3, pp.
  596--615, Jul. 1987.

\bibitem{Internet2}
R.~Summerhill, ``The new internet2 network,'' in \emph{6th GLIF Meeting}, 2006,
  available at
  \url{http://www.internet2.edu/products-services/advanced-networking}.

\bibitem{Mahajan2002}
\BIBentryALTinterwordspacing
R.~Mahajan, N.~Spring, D.~Wetherall, and T.~Anderson, ``Inferring link weights
  using end-to-end measurements,'' in \emph{Proceedings of the 2nd ACM SIGCOMM
  Workshop on Internet measurment}.\hskip 1em plus 0.5em minus 0.4em\relax ACM,
  2002, pp. 231--236. [Online]. Available:
  \url{http://research.cs.washington.edu/networking/rocketfuel/}
\BIBentrySTDinterwordspacing

\bibitem{Wang96}
Z.~Wang and J.~Crowcroft, ``{Quality-of-Service Routing for Supporting
  Multimedia Applications},'' \emph{IEEE Journal on Selected Areas in
  Communications}, vol.~14, no.~7, pp. 1228--1234, 1996.

\bibitem{petre2009}
I.~Petre and A.~Salomaa, ``{Algebraic Systems and Pushdown Automata},'' in
  \emph{{Handbook of Weighted Automata}}, M.~Droste, W.~Kuich, and H.~Vogler,
  Eds.\hskip 1em plus 0.5em minus 0.4em\relax Springer Publishing Company,
  Incorporated., 2009, ch.~7, pp. 257--290.

\end{thebibliography}

\appendices

\section{Polynomial algorithms for path computation in multi-layer networks}
\label{appendix:algos}
The sequence of protocols involved in a feasible multi-layer path is a context-free language. Based on this fact,  Lamali~\textit{et~al.}~\cite{Lamali2013} used automata and language theory tools to compute the shortest feasible path in hops or in adaptation functions. We improve their algorithm in order to compute the shortest path according to any additive metric. We also substantially reduce its complexity.
\label{sec:wpda}
\subsection {Theoretical language aspects of multi-layer paths}

Considering a path $\P = Sf_0U_1f_1U_2f_2\dots U_nf_nD$, let $H_\P=f_1\dots f_n$ denotes the sequence of adaptation functions along $\P$. Let define as an alphabet the set $\overline{\alphabet}=\{\overline{a} \mid a\in\alphabet\}$ and the set $\underline{\alphabet}=\{\underline{a}\mid a\in\alphabet\}$. 

$\trace_\P = x_1\dots x_{n+1}$ is the sequence of protocols used along path $\P$. It is called the \emph{trace} of $\P$. For each $x_i$:
	\begin{itemize}
		\item $x_i = a$ and $x_{i+1}=b$, $\overline{b}$ or $\underline{b}$ means that $U_i$ converts protocol $a$ into $b$ ( $a,b,\overline{b}, \underline{b} \in {\alphabet} \cup \overline{\alphabet}\cup \underline{\alphabet}$) 
		\item $x_i = \overline{a}$ and $x_{i+1} = b$, $\overline{b}$ or $\underline{b}$ means that $U_i$ encapsulates protocol $a$ in $b$
		\item $x_i = \underline{a}$ and $x_{i+1} = b$, $\overline{b}$ or $\underline{b}$ means that $U_i$ decapsulates protocol $b$ from $a$.
	\end{itemize}
	
Here, some additional definitions are needed.	
The set of protocol conversions available on node $U$ is denoted by ${\cal CO}(U)$. The set of encapsulations available on node $U$ is denoted by ${\cal EN}(U)$ and the set of decapsulations  available on node $U$ is denoted by ${\cal CO}(U)$.

	$In(U)$ (resp. $Out(U)$) is the set of protocols that node $U$ can receive (resp. send). More formally: 	
	\begin{itemize}
	\item If $(a\rightarrow b)\in {\cal CO}(U)$ then $a\in In(U)$ and $b\in Out(U)$
	\item If $(a\rightarrow ab)\in {\cal EN}(U)$ then $a\in In(U)$ and $b\in Out(U)$
	\item If $\overline{(a\rightarrow ab)}\in {\cal DE}(U)$ then $b\in In(U)$ and $a\in Out(U)$
	\end{itemize}
	
Obviously, several paths can have the same trace. The set of traces of the feasible paths in a network $\N$ is a context-free language but it is not regular as the encapsulations and decapsulations should be balanced. In fact, it is a well-parenthesized language, and thus requires a stack to be recognized and computed. PDAs are the classical tools to recognize context-free languages. Using weighted PDAs allows associating a weight to each link and adaptation function in order to model any additive metric.
\subsection{Definition of WPDA}
A weighted PDA (WPDA) is a $8$-tuple
$\automaton =(\States, \Alphabet, \Symbols, \Transitions, Q_0, \startsymb, \Fstates, \omega)$
where $\States$ is the set of states, $\Alphabet$ is the input alphabet, $\Symbols$ is the stack symbol set (i.e., stack alphabet) not necessarily different from $\Alphabet$, $\Transitions$ is the set of transitions, $Q_0$ is the initial state, $\startsymb$ is the initial stack symbol, $\Fstates$ is the set of final (accepting) states and $\omega$ is a weight function over the set of transitions (i.e., $\omega :\ \Transitions \rightarrow \Re_+$).

A transition $t\in \Transitions$ is denoted by $t=(Q_i,\langle x, \alpha, \beta\rangle,Q_j)$, where $Q_i$ is the state of $\automaton$ before the transition, $Q_j$ is the state after the transition, $x\in\Alphabet\cup\{\epsilon\}$ is an input symbol, $\alpha\in\Symbols$ is the symbol which is popped from the top of the stack, and $\beta\in\Gamma^*$ is the symbol sequence which is pushed on the top of the stack.

{\bf Remark.} WPDAs are more often formalized as $6$-tuples $\automaton =(\States, \Symbols, \mathcal{M}, q_0, \startsymb, \Fstates)$ where $\mathcal{M}$, called the \textit{Push-Down transition matrix}, is a matrix over a semiring of formal power series. The input alphabet $\Alphabet$, the transitions set $\Transitions$ and the weight function $\omega$ are expressed by a single entity $\mathcal{M}\in((\mathcal{R}\langle\langle\Alphabet^*\rangle\rangle)^{\States\times\States})^{\Symbols^*\times\Symbols^*}$, where $\mathcal{R}\langle\langle\Alphabet^*\rangle\rangle$ denotes the collection of all power series from $\Alphabet^*$ into a semiring $\mathcal{R}$. For simplification purposes, we opted for defining a WPDA as a classical PDA with a weight function over the transition set. For the theoretical foundations of WPDAs, the interested reader can refer to~\cite{petre2009}.

\subsection{From the graph to the WPDA}
Algorithm~\ref{algo:net_pda} converts a multi-layer network $\N$ with a specified pair of nodes $(S,D)$ into a WPDA $\automaton = (\States, \Alphabet, \Symbols, \Transitions, Q_0, \startsymb, \Fstates={Q_F}, \omega)$. 

Computing a feasible path requires to know the current protocol and the last encapsulated one (in order to know if a decapsulation can be performed). Thus Algorithm~\ref{algo:net_pda} creates a state $U_x$ for each node $U$ and each protocol $x\in In(U)$. Being in state $U_x$ indicates that the current protocol is $x$. The last encapsulated protocol is the one on the top of the stack.

The conversion functions $(x\rightarrow y)$ between node $U$ and node $V$ are turned into transitions $(U_x, \langle x, \alpha, \alpha\rangle, V_y)$ in the WPDA. The encapsulation functions $(x\rightarrow xy)$ are converted into pushes of $x$ on the stack $(U_x, \langle \overline{x}, \alpha, x\alpha\rangle, V_y)$ and the decapsulation functions into pops of $x$ from the stack $(U_y, \langle \underline{y},x, \emptyset\rangle, V_x)$.
\begin{algorithm}[!]
\caption{Convert a network into a WPDA}
\label{algo:net_pda}
\begin{algorithmic}
\Require A network $\N=(\Gg=(\V,\E), \alphabet, \AD, h)$, a source $S$ and a destination $D$
\Ensure A WPDA $\automaton= (\States, \Alphabet, \Symbols, \Transitions, Q_0, \startsymb, \{Q_F\}, \omega)$

\State $\Alphabet \gets \alphabet \cup \overline{\alphabet}\cup \underline{\alphabet} $ ; $\Gamma \gets \alphabet \cup \{ \startsymb\} $ 
\State Create $\States$ (the set of states of the WPDA) according to Procedure~\ref{algo:crea_states}
\State Build the transition set $\Transitions$ according to:
\State \hspace{6mm} \--- Procedure~\ref{algo:conv_trans} for the set of conversion functions
\State \hspace{6mm} \--- Procedure~\ref{algo:encap_push} for the set of encapsulation functions
\State \hspace{6mm} \--- Procedure~\ref{algo:decap_pop} for the set of decapsulation functions

\end{algorithmic}
\end{algorithm}


\setcounter{algorithm}{0}
\floatname{algorithm}{Procedure}

\begin{algorithm}[!]
\caption{Create $\States$, the set of states of the WPDA}
\label{algo:crea_states}
\begin{algorithmic}
\State Create a single state $Q_0$ corresponding to node $S$
\State Create a fictitious final state $Q_F$
\State For each node $U\neq S$ in $\nodes$, for each protocol $x \in In(U)$, create a state $U_x$ 
\For{each state $U_x$ s.t. $(S,U)\in \E$, for each $x\in Out(S)$}
\State Create the transition $t=(Q_0, \langle\epsilon, Z_0, Z_0\rangle, U_x)$
\State $\omega(t) \gets 0$
\EndFor
\For{each $x\in In(D)$}
\State Create the transition $t=(D_x, \langle x, Z_0, \emptyset\rangle, Q_F)$
\State $\omega(t)\gets 0$
\EndFor
\end{algorithmic}
\end{algorithm}

\begin{algorithm}
\caption{Transform the conversions}
\label{algo:conv_trans}
\begin{algorithmic}
\For{each link $(U,V)\in \E$ s.t. $U\neq S$}
\For{each $(x\rightarrow y)\in {\cal CO}(U)$}

\If{$y\in In(V)$}
	\For{all $\alpha\in\Symbols$}
	\State Create the transition $t=(U_x, \langle x, \alpha, \alpha\rangle, V_y)$
	\State $\omega(t)\gets h(U,(x\rightarrow y),V)$
	\EndFor

\EndIf
\EndFor
\EndFor
\end{algorithmic}
\end{algorithm}

\begin{algorithm}
\caption{Transform the encapsulations}
\label{algo:encap_push}
\begin{algorithmic}
\For{each link $(U,V)\in \E$ s.t. $U\neq S$}
\For{each $(x\rightarrow xy)\in {\cal EN}(U)$}

\If{$y\in In(V)$}
	\For{all $\alpha\in\Symbols$}
	\State Create the transition $t=(U_x, \langle \overline{x}, \alpha, x\alpha\rangle, V_y)$
	\State $\omega(t)\gets h(U,(x\rightarrow xy),V)$
	\EndFor

\EndIf
\EndFor
\EndFor
\end{algorithmic}
\end{algorithm}

\begin{algorithm}[!]
\caption{Transform the decapsulations}
\label{algo:decap_pop}
\begin{algorithmic}
\For{each link $(U,V)\in \E$ s.t. $U\neq S$}
\For{each $\overline{(x\rightarrow xy)}\in {\cal DE}(U)$}
\If{$x\in In(V)$}
	\State Create the transition $t=(U_y, \langle \underline{y},x, \emptyset\rangle, V_x)$
	\State $\omega(t)\gets h(U,\overline{(x\rightarrow xy)},V)$

\EndIf
\EndFor
\EndFor
\end{algorithmic}
\end{algorithm}

\floatname{algorithm}{Algorithm}
\setcounter{algorithm}{1}

\noindent{\bf Complexity of Algorithm~\ref{algo:net_pda}.} The complexity of Algorithm~\ref{algo:net_pda} is in $O(|\alphabet|^3\times|\E|)$. The number of states created by Procedure~\ref{algo:crea_states} is at worst $2+|\alphabet|\times(|\V|-1)$, and the complexity of Procedure~\ref{algo:crea_states} is in $O(|\alphabet|\times|\V|)$. The number of transitions created by Procedure~\ref{algo:conv_trans} and by Procedure~\ref{algo:encap_push} is in $O(|\alphabet|^3\times|\E|)$, which is also an upper bound for their complexity. The complexity of Procedure~\ref{algo:decap_pop} is bounded by $O(|\alphabet|^2\times|\E|)$. 

\begin{proposition}
\label{prop:equi}
A path $\P$ in a network $\N$ is feasible if and only if its trace $\trace_\P$ is accepted by $\automaton$. 
\end{proposition}

\begin{proof}
Consider a feasible path $\P=Sf_0U_1f_1U_2f_2\dots U_nf_nD$. By construction, for each $3$-tuple $(U_i,f_i,U_{i+1})$ there is a transition:
\begin{itemize} 
\item $t=({(U_i)}_x,\langle x,\alpha,\alpha\rangle,{(U_{i+1})}_y)$ if $f_i=(x\rightarrow y)$
\item $t=({(U_i)}_x,\langle \overline{x},\alpha,x\alpha\rangle,{(U_{i+1})}_y)$ if $f_i=(x\rightarrow xy)$
\item $t=({(U_i)}_x,\langle \underline{x},y,\emptyset\rangle,{(U_{i+1})}_y)$ if $f_i=\overline{(y\rightarrow yx)}$
\end{itemize}
This transition recognizes the $i$-th letter of the trace $\trace_\P$. It is easy to show by induction that $\trace_\P$ is accepted by the automaton.

Conversely, if a trace $\trace_\P$ is accepted by a transition sequence $t_1\dots t_n$ where each $t_i=({(U_i)}_x,\langle x,\alpha,\beta\rangle,{(U_{i+1})}_y)$. Then there is an adaptation function:
\begin{itemize}
\item $f_i=(x\rightarrow y)\in {\cal CO}(U_i)$ if $t_i=({(U_i)}_x,\langle x,\alpha,\alpha\rangle,{(U_{i+1})}_y)$
\item $f_i=(x\rightarrow xy)\in {\cal EN}(U_i)$ if $t_i=({(U_i)}_x,\langle \overline{x},\alpha,x\alpha\rangle,{(U_{i+1})}_y)$
\item $f_i=\overline{(y\rightarrow yx)}\in {\cal DE}(U_i)$ if $t_i=({(U_i)}_x,\langle \underline{x},y,\emptyset\rangle,{(U_{i+1})}_y)$
\end{itemize}
Thus the path $Sf_0U_1f_1U_2f_2\dots U_nf_nd$ is feasible in $\N$.
\end{proof}


The weight of a path $\P=Sf_0U_1f_1U_2f_2\dots U_nf_nD$ is defined as the sum of the weights of its links and its adaptation functions. It is denoted by $h(\P)\overset{def}{=}\sum_{i=1}^nh(U_i,f_i,U_{i+1})$ with $U_{n+1}=D$. 

We define the weight of a transition sequence as the sum of the weights of each transition (i.e., $\omega(\{t_1,t_2,\dots,t_n\})=\sum_{i=1}^n\omega(t_i)$). The weigh of a word $w$, denoted by $\omega(w)$, is the weight of the transitions that accept $w$ in $\automaton$. But as $\automaton$ may be nondeterministic, it is possible that several transition sequences accept the same word. Thus we consider only the sequence of transitions of minimum weight that accepts $w$. More formally, $\omega(w)=\min_{t_1,\dots, t_n\in \Transitions}\omega(\{t_1,\dots t_n\})$ s.t. $\{t_1,\dots t_n\}$ accepts $w$.
\begin{lemma}
If $\automaton$ accepts the trace $\trace_\P$ of a path $\P$, then $\omega(\trace_\P)=h(\P^*)$, where $\P^*$ is the path of minimum weight having $\trace_\P$ as trace.
\end{lemma}
\begin{proof}
By definition, $\omega(\trace_\P)=\omega(\{t_1,\dots,t_n\})$, where $\{t_1,\dots,t_n\}$ is the transition sequence with minimal weight which accepts $\trace_\P$. From $\{t_1,\dots,t_n\}$, it is possible to build the path $\P^*$ (inversing the conversion in Algorithm~\ref{algo:net_pda}) such that $\trace_{\P^*}=\trace_\P$ and $h(\P^*)=\omega(\trace_\P)$.

Suppose that $\exists \P'$ s.t. $\trace_{\P'}=\trace_\P$ and $h(\P')<\omega(\trace_\P)$, then it is possible to build from $\P'$ a sequence of transitions that corresponds to the links and adaptation functions involved in $\P'$ (as in Algorithm~\ref {algo:net_pda}). Let this sequence be $t_1'\dots,t_n'$. The weight of each transition $t_i'$ corresponds to the weight of an adaptation function associated to a link in $\P'$. The weight of $t_1'\dots,t_n'$ is then less than $\omega(\trace_\P)$, and by Proposition~\ref{prop:equi}, this sequence accepts $\trace_\P$. This is inconsistent with the definition of $\omega(\trace_\P)$.
\end{proof}

\subsection{Computing the minimal weight trace} 
\label{sec:short_word}
In order to compute the minimum weight trace and its corresponding path, $\automaton$ is converted into a weighted Context-Free Grammar (WCFG). 

\subsubsection{From the WPDA to a WCFG} 
\label{sec:pda-cfg}
A WCFG is a CFG with a weight function over the set of production rules. The conversion of a PDA into a CFG is well-known. The conversion of a WPDA into a WCFG is done in the same way, in addition the weight of each transition is assigned to the corresponding production rules (called rules in Algorithm~\ref{WPDA-WCFG}) in the WCFG.

Algorithm~\ref{WPDA-WCFG} is an adaptation of the general method described in~\cite{Hop06}. It converts $\automaton$ into a WCFG $\G= ({\mathcal Q},\Sigma, [Q_0],{\mathcal R},\pi)$ where:
\begin{itemize}
\item ${\mathcal Q}$ is the set of nonterminals,
\item $\Sigma$ is the alphabet or set of terminals (the same as the WPDA input alphabet),
\item $[Q_0]$ is the initial symbol (initial nonterminal, or axiom),
\item ${\mathcal R}$ is the set of production rules,
\item $\pi : {\mathcal R}\rightarrow \Re_+$ is the weight function over the set of production rules.
\end{itemize}

\begin{algorithm}
\caption{Convert a WPDA into a WCFG}
\label{WPDA-WCFG}
\begin{algorithmic}
\Require $\automaton = (\States, \Alphabet, \Symbols, \Transitions, Q_0, \startsymb, \{Q_F\}, \omega)$
\Ensure $\G= ({\mathcal Q},\Sigma, [Q_0],{\mathcal R},\pi)$
\State Create the axiom $[Q_0]$
\For{each state $U_x \in \States$}
	\State Create the nonterminal $[Q_0Z_0U_x]$ 
	\State Create the rule $[Q_0] \rightarrow [Q_0Z_0U_x]$
\EndFor
\For{each transition $(U_x,\langle x,\alpha,\beta \rangle,V_y)$}
	\If{$\beta=\emptyset$ (pop)}
		\State Create a nonterminal $[U_x\alpha V_y]$ 
		\State Create the rule $r = [U_x\alpha V_y] \rightarrow x$
		\State $\pi(r)\gets \omega(U_x,\langle x,\alpha,\emptyset \rangle,V_y)$
	\EndIf
	\If{$\beta=\alpha$ (conversion transition)}
		\For{each $Q_i\in \States$}
		\State Create nonterminals $[U_x\alpha Q_i]$ and $[V_y\alpha Q_i]$
		\State Create the rule $r=[U_x\alpha Q_i] \rightarrow x[V_y\alpha Q_i]$
		\State $\pi(r)\gets \omega(U_x,\langle x,\alpha,\beta\rangle,V_y)$
		\EndFor
	\EndIf
	\If{$\beta=x\alpha, \ x \in \Gamma$ (push)}
	 	\For{each $(Q_i,Q_j)\in {\States}^2$}
		\State Create nonterminals $[U_x\alpha Q_j]$, $[V_y\alpha Q_i]$ and $[Q_i\alpha Q_j]$
		\State Create the rule $r=[U_x\alpha Q_j] \rightarrow x[V_yxQ_i][Q_i\alpha Q_j]$
		\State $\pi(r)\gets \omega(U_x,\langle x,\alpha,x\alpha\rangle,V_y)$
		\EndFor
	\EndIf
\EndFor
\end{algorithmic}
\end{algorithm}
\noindent{\bf Complexity of Algorithm~\ref{WPDA-WCFG}.} The number of nonterminals is bounded by $O(|\Gamma|\times|\States|^2)$ (as each nonterminal is in the form $[Q_ixQ_j]$ with $Q_i,Q_j\in\States$ and $x\in \Gamma$. The number of production rules is bounded by $O(|\Transitions|\times|\States|^2)$. Thus the worst case complexity of Algorithm \ref{WPDA-WCFG} is bounded by $O(|\delta|\times|\States|^2)$. This corresponds to $O\left(|\alphabet|^5\times|\V|^2\times|\E|\right)$.

\subsubsection{The minimum weight derivation tree}
\label{sec:min_trace}
Generating the minimum weight trace (and then the minimum weight path) requires to build its derivation tree. Let $[X]$ be a nonterminal, we define $\ell([X])$ as the sum of the weights of the productions needed for, starting from $[X]$, deriving a word in $\Alphabet^*$. Thus $\ell([Q_0])$ is the weight of the minimum weight trace.

The function is $\ell :{\{{\mathcal Q}\cup \Sigma \cup \{\epsilon\}\}}^*\rightarrow \mathbb{N}\cup\{\infty\}$ s.t.:
\begin{itemize}
\item if $w = \epsilon$ or $w \in\Alphabet$ then $\ell(w) = 0$,
\item if $w=\alpha_1\dots\alpha_n$ (with $\alpha_i \in \{{\mathcal Q}\cup \Sigma \cup \{\epsilon\}\}$) then $\ell(w) = \sum_{i=1}^n\ell(\alpha_i)$.
\item Let $r_1=[X]\rightarrow \gamma_1,r_2=[X]\rightarrow \gamma_2,\dots,r_k=[X]\rightarrow\gamma_k$ be the set of production rules having $[X]$ as left part. Then $\ell([X])=\min\{\pi(r_1)+\ell(\gamma_1),\dots,\pi(r_k)+\ell(\gamma_k)\}$
\end{itemize}

Knuth's algorithm~\cite{Knuth77} can be adapted to compute the minimum weight derivation tree of a grammar. This corresponds to the weight of $\trace_\P$, where $\P$ is the shortest path to compute. The adapted algorithm maintains a list of production rules and updates the $\ell[X]$ according to the formula above. The sketch of the algorithm is as follows:
\begin{itemize}
\item Initialize $\ell([X])$ to $\infty$ for each nonterminal $[X]$
\item  For each production rule $[X]\rightarrow \alpha_1\dots \alpha_n$ update $\ell([X])$ as follows:
\\$\ell([X]) \gets \min\{\ell([X]),\pi(r)+\sum_{i=1}^n\ell(\alpha_i)\}$
\end{itemize}
The algorithm terminates when all the $\ell[X]$ have the right value and no additional update is possible. Implementing this algorithm with Fibonacci heaps leads to a $O(|\mathcal{Q}|\log|\mathcal{Q}|+\mathcal{|R|})$ complexity~\cite{Tarjan1987}, which corresponds to $O(|\alphabet|^5\times|V|^2\times|\E|)$.

With the correct values of $\ell[X]$, it is trivial to generate the word with the minimum weight derivation tree.

\subsection{Deriving the shortest path from its trace}
\label{sec:word-path}
Algorithm \ref{find-C} is a generalization of an algorithm proposed in~\cite{Lamali2013}.
It takes as input the minimum weight trace $\trace_{\P}$ accepted by $\automaton$ and computes the path $\P$ that matches it\footnote{It is possible that several paths match the trace. In this case the path can be chosen randomly or according to a load-balancing policy.}. 

Algorithm \ref{find-C} starts on $nodes[1]=S$ then checks at each step all the links in $\E$ which match the current letter (protocol) in $\trace_\P$. If $\trace_\P = x_1x_2\dots x_n$ $(x_i \in \alphabet\cup\overline{\alphabet}\cup\underline{\alphabet})$, then at each step $i$, the algorithm starts from each node $U$ in $nodes[i]$ and adds to $links[i]$ all the links $(U,V)$ which match $x_i$. Each $V$ is added in $nodes[i+1]$. The value $weights[(U,V),i]$ is the cost of using link $(U,V)$ at step $i$. It corresponds to the weight $h(U,f_i,V)$ where $f_i$ is the adaptation function used at step $i$. When the trace $\trace_\P$ is completely covered, a classical shortest path algorithm from $S$ to $D$ in the graph $(nodes, links, weights)$ computes the minimum weight path.
\begin{algorithm}
\caption{Computing the shortest path}
\label{find-C}
\begin{algorithmic}
\Require{The network $\N$ and $\trace_\P$}
\Ensure{The shortest path $\P$}
\State $nodes[1] \gets S$ ; $i \gets 2$
\While{The trace is not completely covered}
\For{each $U \in nodes[i]$, $V \in\nodes \ s. t.\ (U,V) \in \E$}
	\If{$x_i\in\alphabet$, $x_i\in Out(U)$, $x_i\in In(V)$ and $(x_{i-1}\rightarrow x_i)\in {\cal CO}(U)$}
		\State Add $(U,V)$ in $links[i]$ and $V$ in $nodes[i+1]$
		\State $weights[(U,V),i]\gets h(U,(x_{i-1}\rightarrow x_i),V)$
	\EndIf
	\If{$x_i\in\overline{\alphabet}$, $x_i\in Out(U)$, $x_i\in In(V)$ and $(x_{i-1}\rightarrow x_{i-1}x_i)\in {\cal EN}(U)$}
		\State Add $(U,V)$ in $links[i]$ and $V$ in $nodes[i+1]$
		\State $weights[(U,V),i]\gets h(U,(x_{i-1}\rightarrow x_{i-1}x_i),V)$
	\EndIf
	\If{$x_i\in\underline{\alphabet}$, $x_i\in Out(U)$, $x_i\in In(V)$ and $\overline{(x_i\rightarrow x_ix_{i-1})}\in {\cal DE}(U)$}
		\State Add $(U,V)$ in $links[i]$ and $V$ in $nodes[i+1]$
		\State $weights[(U,V),i]\gets h(U,\overline{(x_i\rightarrow x_i x_{i-1})},V)$
	\EndIf
\EndFor
\State $i++$
\EndWhile
\State Compute The shortest path from $S$ to $D$ in $(nodes, links)$
\end{algorithmic}
\end{algorithm}

\noindent{\bf Complexity of Algorithm \ref{find-C}.} The complexity of Algorithm \ref{find-C} is bounded by $O(|\trace_\P|\times|{\mathcal V}|\times|\E|)$ in the worst case.

\section{Proof that SYM-HAM is $\mathsf{NP}$-complete}
\label{app:np-complete}
{\bf Problem SYM-HAM.} Given a directed symmetric graph $\Gg=(\V,\E)$ and a pair of nodes $(S,D)$, is there a Hamiltonian path from $S$ to $D$ in $\Gg$?

\begin{proposition}
SYM-HAM is $\mathsf{NP}$-complete.
\end{proposition}

\begin{proof}
First, it is clear that SYM-HAM is in $\mathsf{NP}$. Thus, we prove its $\mathsf{NP}$-hardness by providing a polynomial reduction from the Hamiltonian path problem in \emph{undirected} graphs to SYM-HAM. Consider an undirected graph $\Hh=(\V',\E')$ and a pair of nodes $(S',D')$. It is $\mathsf{NP}$-complete to know whether there is an undirected Hamiltonian path between $S'$ and $D'$. The reduction builds an instance of SYM-HAM as follows:
 A graph $\Gg=(\V,\E)$ where $\V=\V'$. For each undirected edge $(U,V)$ in $\Hh$, create the directed links $(U,V)$ and $(V,U)$ in $\Gg$.

Let $\P'=SU_1U_2\dots U_nD$ be a Hamiltonian path in $\Hh$. For each edge $(U_i,U_{i+1})$ in $\Hh$, one can take the corresponding directed link $(U_i,U_{i+1})$ in $\Gg$ and construct a Hamiltonian path in $\Gg$.

Now let $\P'=SU_1U_2\dots U_nD$ be a (directed) Hamiltonian path in $\Gg$. By replacing each link $(U_i,U_{i+1})$ by the corresponding undirected edge (in $\Hh$), one obtains a path visiting all the nodes exactly once in $\Hh$ (as $\Gg$ and $\Hh$ have the same set of nodes). Thus, the obtained path is a Hamiltonian path in $\Hh$.

So $\Hh$ admits an undirected Hamiltonian path between $S'$ and $D'$ if and only if $\Gg$ admits a directed Hamiltonian path from $S$ to $D$.
\end{proof}


\end{document}